\documentclass[%
 reprint,
 amsmath,amssymb,
 aps,
]{revtex4-2}

\usepackage{amsfonts}
\usepackage{amsthm}
\usepackage{dsfont}
\usepackage{amsmath,amssymb}
\usepackage{graphicx}
\usepackage{dcolumn}
\usepackage{bm}

\newcommand{\ud}{\mathrm{d}}
\newcommand{\ui}{\mathrm{i}}
\newcommand{\ue}{\mathrm{e}}

\providecommand{\norm}[1]{\lVert#1\rVert}

\providecommand{\abs}[1]{\lvert#1\rvert}

\newcommand{\tr}{\operatorname{tr}}
\newcommand{\pa}{\partial}
\newcommand{\la}{\langle}
\newcommand{\ra}{\rangle}

\newcommand{\R}{\mathds{R}}

\newcommand{\N}{\mathds{N}}
\newcommand{\C}{\mathds{C}}

\renewcommand{\Im}{\operatorname{Im}}
\renewcommand{\Re}{\operatorname{Re}}

\newcommand{\range}{\operatorname{range}}

\newtheorem{thm}{Theorem}[section]
\newtheorem{lem}[thm]{Lemma}
\newtheorem{prop}[thm]{Proposition}

\newtheorem{Def}[thm]{Definition}

\begin{document}


\title{Decoherence Time Scales and the H{\"o}rmander condition}

\author{Roman Schubert}

\author{Thomas Plastow}%

\affiliation{
School of Mathematics, University of Bristol, Bristol, United Kingdom 
}

\date{\today}

\begin{abstract}
We consider an open quantum system described by the GKLS equation and we are interested in the onset of decoherence. We are in particulary interested in situations where only some degrees of freedom of the system are coupled to the environment, and we want to understand if, and how fast, the noise travels through the system and eventually affects all degrees of freedom. We find that this can be understood in terms of the H{\"o}rmander condition, a condition on the commutators of the Hamiltonian vectorfields of the Lindbald operators and the internal Hamiltonian, which is a condition for hypoellipticity known from the theory of PDE's.  We show that for Gaussian quantum channels this condition leads to a delay in the onset of decoherence and can as well be used to detect decoherence free subsystems.

\end{abstract}

\maketitle

\section{\label{sec:Introduction}Introduction}

We will consider open quantum systems whose time evolution can be described by the GKLS equation \cite{Lind76,GKS76}. 
An open quantum systems is a quantum system which is interacting with an environment, and in many situations when the coupling to the environment is weak and 
memory effects in the environment can be neglected, the evolution of the density operator $\hat\rho$ is subject to the GKLS equation \cite{BrePet02,AliLen07}
\begin{equation}\label{eq:Lindblad}
\ui\hbar \pa_t \hat\rho=[\hat H,\hat\rho]+\frac{\ui}{2}\sum_{k=1}^K 2\hat L_k\hat \rho\hat L_k^{\dagger}-\hat L_k^{\dagger}\hat L_k\hat \rho -\hat \rho \hat L_k^{\dagger}\hat L_k\,\, .
\end{equation}
Here  the Hamiltonian $\hat H$ is Hermitian and describes the internal dynamics of the system, and the Lindblad operators $\hat L_k$ describe the influence of the environment on the system. Typical examples for Lindblad operators are 
(i) $\hat L_k=\sqrt{\Lambda} \, \hat q_k$, $k=1, \cdots , d$, where $d$ is number of degrees of freedom,  which are used if the environment can be modelled by random scatterers,  or (ii) $L_k=\gamma_k^{(-)} \hat a_k$ and 
$L_{k+d}= \gamma_k^{(+)} \hat a_k^{\dagger}$, $k=1, \cdots, d$, where $\hat a_k^{\dagger}, \hat a_k$ are creation and annihilation operators of the $k$'th mode respectively, which are used to model coupling to a heat bath, \cite{BrePet02,JoosEtAl03,AliLen07,Schlo19}.   

Decoherence is the suppression of interference effects due to the 
influence of the environment, \cite{JooZeh85,Zur91,JoosEtAl03,Horn09,Schlo19,Zur03}, it is an effect which typically sets in on very short time scales and it is a major obstacle for the practical implementation of    
quantum computing, as the superposition principle is the main resource in quantum information.  In this paper we will consider decoherence  for continuous variable quantum systems, see  \cite{BrauvLo05,Weed12,AdeRagLee14}, and in particular we will be interested in the situation where the environment is only coupled to some degrees of freedom and the internal dynamics is needed to transport the effect of the noise through the whole system. We will give a general condition 
which characterises the situations where decoherence eventually affects the whole system. Furthermore the methods we present can as well be used identify subsystems on which the onset of decoherence is delayed and decoherence free subsystems, related to \cite{LidWha03,lid14,Yam14}. 

 The Hilbert space of the systems we study  is given by $L^2(\R^n)$ and we will assume that the operators $\hat \rho, \hat H$ and $\hat L_k$ are given  as 
 Weyl-quantisations (see \cite{DimSjo99,Zwo12})
of phase space functions $\rho(x), H(x)$ and $L_k(x)$, where $x=(p,q)\in \R^{n}_p\oplus\R^n_q$ denote momentum and position  in phase space. We will focus on a class of systems for which the time evolution can be solved explicitly, namely we will assume that 
\begin{equation}
    H(x)=\frac{1}{2} x\cdot Q x\quad\text{and} \quad L_k(x)=x\cdot\Omega^T l_k
\end{equation}
Here  $Q$ is symmetric and real valued, $\Omega=\begin{pmatrix} 0 & -I\\ I & 0\end{pmatrix}$ and $l_k\in\C^{2n}$, the peculiar parametrization of $L_k$ with $\Omega$ is chosen so that the Hamiltonian vectorfield of $L_k$, $X_{L_K}=\Omega\nabla_xL_k$, is given by $l_k$. As the examples after equation \eqref{eq:Lindblad} show, this includes typical examples of Lindblad operators used in applications. 
 In this case the Lindblad equation can be rewritten as an equation for the Wignerfunction of $\hat{\rho}(t)$, $\rho(t,x)$, 
\begin{equation}\label{eq:Lindblad-ps}
\pa_t\rho=X_0 \rho+\nabla X_0 \, \rho +\frac{\hbar}{2}\sum_{k=1}^{2K} X_k^2\rho\,\, ,
\end{equation}
where the vector fields $X_k$, $k=0,1, \cdots , 2K$  are given by 
\begin{equation}\label{eq:def_X0}
\begin{split}
    X_0=&-(Fx)\cdot \nabla_x\\
&+\sum_k\Im L_k(x) \Re l_k\cdot\nabla_x-\Re L_k(x)\Im l_k\cdot\nabla_x\,\, ,
\end{split}
\end{equation}
with $F=\Omega Q$,  and for $k=1, 2, \cdots , K$ 
\begin{equation}\label{eq:def_Xk}
X_k =\Re l_k\cdot\nabla_x\,\, ,\quad X_{K+k}=\Im l_k\cdot\nabla_x\,\, .
\end{equation}

The equation \eqref{eq:Lindblad-ps} has a transport part, $X_0$, and a diffusive part given by the sum of squares of $X_k$, $k=1, 2, \cdots , 2K$. Such equations occur in the 
description of stochastic processes, and this connection is not surprising as we are in a situation where we treat the influence of the environment as noise. 
In the theory of stochastic processes it is important to understand under which conditions on the vector fields $X_0,X_k$  the solutions to \eqref{eq:Lindblad-ps} 
are smooth for $t>0$, even for singular initial conditions, this is a condition which is called hypoellipticity in the theory of Partial Differential Equations, \cite{Bram14}. 
There is a celebrated result by H{\"o}rmander \cite{Hor67} which gives a criterion for hypoellipticity, it is formulated in terms of the commutators of the vector fields 
$X_0,X_k$, $k=1,2,\cdots $. 

\begin{Def}\label{def:Hormander}
Let us consider the  subspaces $V_k\subset \R^n_p\oplus\R^n_q$, $k=0,1,2, \cdots $,  spanned by the $X_j$, $j=1,2,\cdots $ and iterated commutators with $X_0$, 
\begin{align}
V_0:=&\la X_j\, ;\,\, j=1, \cdots ,2K\ra \label{eq:def-V0} \\
V_k:=&\la Y, [Y,X_0] \, ; \,\, Y\in V_{k-1}\ra\,\, .\label{eq:def-Vk}
\end{align}
We say that  $X_j$, $j=0,1, \cdots , K$, satisfy the {\bf{H{\"o}rmander condition}} if for some $k$ we have $V_k=\R^n_p\oplus\R^n_q$.
\end{Def}

See \cite{Bram14,AgrBarBos20} for more information on the geometry behind this condition. The H{\"o}rmander condition is in control theory and sub-Riemannian geometry sometimes as well called the bracket condition or the Chow condition. 
To explain its meaning let us denote by $\phi_k^t$ the flow generated by $X_k$, then we have, \cite{AgrBarBos20},
 \begin{equation}
 \phi_k^{-t}\circ \phi_{k'}^{-t}\circ \phi_k^{t}\circ \phi_{k'}^{t}=t^2 [X_k,X_{k'}]+O(t^3)\,\, ,
 \end{equation}
 so by combining the flows of the vector fields $X_k$, $k=0,1,\ldots, $, we can move in a direction given by a commutator $[X_k,X_{k'}]$. 
 By iterating this argument one can show that by suitably composition of the flows one can move in the direction of iterated commutators, too. This idea is formalised in the Chow Rashevsky Theorem, see \cite{Bram14,AgrBarBos20}.

 In the phase space representation of the Lindblad equation, \eqref{eq:Lindblad-ps}, we have transport in the direction of the vector field $X_0$ and diffusion in the 
 direction of the vector fields $X_k$, $k=1, 2,\ldots , 2K$, and so in view of the Chow Rashevsky Theorem it is natural to expect that the diffusion will affect all 
 parts of the system if H{\"o}rmander's condition holds. This is particular interesting in situations where the environment couples only to some degrees of freedom of the system, 
 and we would like to understand under which conditions decoherence will affect all parts of the system, or only some parts. H{\"o}rmander's condition will give a sufficient condition for 
 decoherence to spread through the whole system.

In order to illustrate the condition let us take as example a free particle in one-degree of freedom with collisional decoherence, i.e., $H=\frac{1}{2m}p^2$ and $L=\sqrt{\Lambda} q$. Then we have 
\begin{equation}\label{eq:coll-dec}
  X_0=\frac{1}{m}p\pa_q\,\, , X_1=\sqrt{\Lambda}\pa_p \quad\text{and}\quad [X_1,X_0]=\frac{\sqrt{\Lambda}}{m}\pa_q
\end{equation}
and so $V_0=\R_p\subset \R_p\oplus\R_q$ and $V_1=\R_p\oplus\R_q$ and the H{\"o}rmander condition holds.  On the other hand side, if we choose $H=\lambda pq$, a Hamiltonian used for instance as a normal form near an unstable fixed point, and the same Lindblad term $L=\sqrt{\Lambda}q$, then we have 
\begin{equation}
  X_0=\lambda(p\pa_p+q\pa_q)\,\, , X_1=\sqrt{\Lambda}\pa_p \quad\text{and}\quad [X_1,X_0]=\lambda X_1
\end{equation}
so $V_0=\R_p$ and $V_1=V_0$, and the H{\"o}rmander condition does not hold. 
 
In the situation we consider, namely that the internal Hamiltonain is quadratic in $\hat x$ and the Lindblad operators are linear in $\hat x$, the time evolution of the system generated by the GKLS equation is a one parameter semigroup of Gaussian Channels, $\mathcal{V}_t$, \cite{Weed12}.  A Gaussian Channel can be characterised by its action on the characteristic function of a state. Recall that if $\hat\rho$ is a density operator, then its characteristic function is defined as 
 \begin{equation}
 \chi_{\hat\rho}(\xi):=\tr[\hat \rho T(\xi)]\,\, ,
 \end{equation} 
 where $T(\xi)=\ue^{-\frac{\ui}{\hbar} \xi\cdot \hat x}$. The characteristic function is the Fourier transform of the Wigner function, hence  the  Wigner function can be computed from the characteristic function by the inverse Fourier transform. Now the action of a Gaussian Channel $\mathcal{V}$ on the characteristic function is given by 
 \begin{equation}\label{eq:GC}
     \chi_{\mathcal{V}\hat\rho}(\xi)=\chi_{\hat\rho}(R^T\xi)\ue^{-\frac{1}{2\hbar}D(\xi)}
 \end{equation}
where $R:\R^{2n}\to\R^{2n}$ is a linear map and $D(\xi)=\xi^TD\xi$ is a quadratic form in $\xi$, \cite{HolWer01,Weed12}. The map $R$ and the quadratic form $D$ define the Gaussian Channel uniquely, and they have to satisfy the condition 
\begin{equation}
D+i(\Omega-R\Omega R^T)\geq 0
\end{equation}
in order to guarantee complete positivity, \cite{HolWer01}. Therefore the solution to the GKLS equation \eqref{eq:Lindblad-ps} is given in terms of the corresponding maps $R_t$ and quadratic forms $D_t$, and  our strategy to understand the onset of decoherence is to compute the short time behavior of $D_t$, since  decoherence holds if $D_t>0$.

To demonstrate the presence of decoherence we will look at the evolution of off-diagonal terms in cat states. A cat state is  a superposition of two coherent states $|\psi\ra=\frac{1}{\sqrt{2}}(|z_1\ra+|z_2\ra)$, centred at two phase space points $z_1=(p_1,q_1)$ and $z_2=(p_2,q_2)$, respectively. Its density matrix has 4 terms
 \begin{equation}
\hat\rho=\frac{1}{2}\big(|z_1\ra\la z_1|+|z_2\ra\la z_2|+|z_1\ra\la z_2|+|z_2\ra\la z_1|\big)
\end{equation}
and one manifestation of decoherence is the suppression of the off-diagonal terms $|z_i\ra\la z_j|$, $z_i\neq z_j$, see \cite{Zur91,PabHabZur93,Zur03}.  The corresponding characteristic functions are given by 
\begin{equation}\label{eq:char-z-z}
    \chi_{|z_i\ra\la z_j|}(\xi)=\ue^{-\frac{\ui}{\hbar}\bar z_{ij}\cdot \xi}\ue^{-\frac{1}{4\hbar}(\xi-\Omega\delta z_{ij})^2}
\end{equation}
where $\delta z_{ij}=z_i-z_j$ and  $\bar z_{ij}=(z_i+z_j)/2$. If we now apply a Gaussian channel to this term we obtain 
\begin{equation}
\ue^{-\frac{\ui}{\hbar}\bar z_{ij}\cdot R^T\xi}\ue^{-\frac{1}{4\hbar}[(R^T\xi-\Omega\delta z_{ij})^2+2\xi \cdot D\xi]}    
\end{equation}
so if $\delta z_{ij}\neq 0$ and $D>0$ then this term will be small, and this is how decoherence manifests itself on the characteristic function. 
We will use the Hilbert Schmidt norm, 
\begin{equation}
 \norm{\hat\rho}_{HS}:=(\tr[\hat \rho^{\dagger}\hat\rho])^{1/2}\,\, ,   
\end{equation}
as a measure of the size of these terms, this norm is particular convenient since we have 
\begin{equation}\label{eq:HS-chi}
 \norm{\hat \rho}_{HS}^2=\frac{1}{(2\pi\hbar)^n}\int \abs{\chi(\xi)}^2\ud \xi\,\, ,  
\end{equation}
and $\norm{|z_1\ra\la z_2|}_{HS}=1$. 

We can now formulate one of our main results.

 \begin{thm}\label{thm:1} Suppose   that $\hat \rho_0=|z_1\ra \la z_2|$ and set $\delta z=z_1-z_2$, then we have for short times the estimate
 \begin{equation}
 \norm{\hat \rho_t}_{HS} = \ue^{-\frac{1}{2\hbar} [d_0(\delta z) t +O(t^{2})] } (1+O(t)) 
 \end{equation}
 and if $\Omega \delta z\in V_{j-1}^{\perp}$ then the more precise estimate 
 \begin{equation}
 \norm{\hat \rho_t}_{HS} = \ue^{-\frac{1}{2\hbar} [d_j( \delta z) t^{2j+1} +O(t^{2j+2})] } (1+O(t)) 
 \end{equation}
 holds, where for $j=0,1,2, \ldots $
 \begin{equation}
 d_j(\delta z)=\frac{1}{(2j+1) (j!)^2}\sum_{k=1}^{K} \abs{L_k(F^j  \delta z)}^2 \,\, ,
 \end{equation}
 and $F=\Omega Q$ is the Hamiltonian map of $H$. 
 \end{thm}
 
 We will give a slightly more precise formulation in Theorem \ref{thm:decoherence} in Section \ref{sec:Horm}. The proof is based on \eqref{eq:HS-chi} and \eqref{eq:GC}, and  estimates on $D_t$ we will develop in section \ref{sec:Horm}. One can derive similar upper bounds for the Wignerfunction $\rho_t(x)$ based on $\abs{\rho_t(x)}\leq \frac{1}{(2\pi\hbar)^n}\int \abs{\chi_t(\xi)}\, d\xi $. 
 
 For the example \eqref{eq:coll-dec} we find for short times 
     \begin{equation}
         \norm{\hat \rho_t}_{HS} \sim\begin{cases} \ue^{-\frac{1}{2\hbar}\Lambda \abs{\delta q}^2 t} & \delta q\neq 0\\
         \ue^{-\frac{1}{2\hbar}\frac{1}{3m}\Lambda \abs{\delta p}^2 t^3} & \delta q=0 
         \end{cases}
     \end{equation}
     where $\delta z=(\delta p,\delta q)\neq 0$. We see that the onset of decoherence is delayed if $\delta q=0$, this can be explained by observing that the effects of noise on different regions in $q$ are independent of each other, so if $\delta q\neq 0$ the relative phase between the two coherent states will become random. But if $\delta q=0$ and $\delta p\neq 0$, then the internal dynamics is needed to separate  the two coherent states before the noise can randomise the relative phase. This effect was already observed and explained in \cite{PabHabZur93,Zur03}, and a quantitative estimate appeared in \cite{JoosEtAl03}.

 The plan of the paper is as follows. In Section \ref{sec:background} we recall the explicit formulas for the time evolution and in Section \ref{sec:Horm} we use the H{\"o}rmander condition to derive explicit expression for the onset of decoherence at short times. We want to emphasise that the material in Sections \ref{sec:background} and \ref{sec:Horm} are based on similar results in \cite{Kup72,LancPol94} for evolution equations of Fokker-Planck type. Our main contribution is to apply these ideas to the study of decoherence.  In Section \ref{sec:examples} we apply our results to some examples to help illuminate their meaning and in Section \ref{sec:summary} we summarise our results and indicate some future directions of study.

\section{Background}\label{sec:background}

In this section we will recall the explicit form for the time evolution generated by the Lindblad equation in the situation we consider. 
With the notation from \eqref{eq:def_X0} and \eqref{eq:def_Xk} we have 
\begin{equation}\label{eq:X0_lk}
X_0 \rho =-(Ax)\cdot \nabla \rho \quad \text{and}\quad \sum_{k=1}^{2K}  X_k^2 \rho=\nabla \cdot M\nabla \rho \,\, ,
\end{equation}
where 
\begin{equation}
A=F+N\Omega
\end{equation}
 with
\begin{align}
N&=\sum_k\Re l_k\Im l_k^T-\Im l_k\Re l_k^T\,\, ,\label{eq:def-N}\\
M&=\sum_k \Re l_k\Re l_k^T+\Im l_k \Im l_k^T\,\, .
\end{align}
Notice that 
\begin{equation}
\sum_{k} \bar l_k l_k^T=M+\ui N
\end{equation}
and  $M$ is symmetric, and $N$ is anti-symmetric, respectively.

We can now write down an explicit solution to the Lindblad equation, expressed in terms of the characteristic function of $\hat{\rho}(t)$. 

\begin{thm}\label{thm:time-evo}
Set 
\begin{equation}\label{eq:time-evo-R-D}
R_t:=\ue^{tA}\quad\text{and}\quad D_t=\int_0^t R_s MR_s^T\, \ud s\,\, ,
\end{equation}
and suppose $\rho(t,x)$ is a solution to \eqref{eq:Lindblad-ps} with $\rho(0,x)=\rho_0(x)$,  then for $t\geq 0$ 
\begin{equation}\label{eq:time-evo-F}
\chi(t,\xi)=\chi_0(R_t^T\xi)\, \ue^{-\frac{1}{2\hbar} \xi\cdot D_t\xi}\,\, .
\end{equation}
\end{thm}

We can rewrite this result for the Wigner function if we assume $D_t>0$ for $t>0$, then we find by inserting the Fouriertransform that 
\begin{equation}
    \rho(t,x)=\int K(t,x,y)\, \rho_0(y)\, dy\, ,
\end{equation}
where the propagator is given by 
\begin{equation}
    K(t,x,y)=\frac{1}{(2\pi\hbar)^n\sqrt{\det D_t}} \, e^{-\frac{1}{2\hbar}(x-R_ty)D_t^{-1}(x-R_t y)}\,\, .
\end{equation}
This is a classical result for Fokker-Planck type equations. A first special case with a degenerate $M$ goes back to Kolmogorov \cite{Kolm34}, the first time the general case appears seems to be \cite{Kup72}, see as well 
\cite{LancPol94} for a more recent study. In the physics literature see \cite{Ris89} in the context of Fokker-Planck equations, and \cite{BroOzo10} specifically for the Lindblad equation.

\begin{proof} 
In the case of quadratic $H$ and linear $L_k$ the Lindblad equation \eqref{eq:Lindblad-ps} reduces to 
\begin{equation}
\pa_t\rho(t,x)=-(Ax)\cdot \nabla \rho(t,x)-\tr A\, \rho(t,x)+\frac{\hbar}{2}\nabla \cdot M\nabla \rho(t,x)\,\, , 
\end{equation}
which gives for the characteristic function  $\chi(t,\xi)=\int \ue^{-\frac{\ui}{\hbar} x\cdot \xi} \rho(t,x)\, \ud x$ the equation 
\begin{equation}\label{eq:FLindblad}
\pa_t \chi(t,\xi)=(A^T\xi)\cdot \nabla_{\xi}\chi(t,\xi)-\frac{1}{2\hbar} \xi \cdot M\xi \chi(t,\xi)\,\, .
\end{equation}
If we make an Ansatz $\chi(t,\xi)=\chi_0(R_t^T\xi)\ue^{-\frac{1}{2\hbar}\xi \cdot D_t\xi}$, with $R_t$ and $D_t$ $2n\times 2n$ matrices, $D_t$ symmetric, with $R_0=I$ and $D_0=0$, then the left hand side of 
\eqref{eq:FLindblad} is 
\begin{equation}
(\pa_t R_t^T\xi)\cdot (\nabla \chi_0)(R_t^T\xi) \ue^{-\frac{1}{2\hbar}\xi \cdot D_t\xi} -\frac{1}{2\hbar} \xi \cdot \pa_t D_t\xi \chi(t,\xi)
\end{equation}
whereas the right hand side gives
\begin{equation}
\begin{split}
(A^T\xi)&\cdot (R_t\nabla \chi_0)(R_t^T\xi)\ue^{-\frac{1}{2\hbar}\xi \cdot D_t\xi}\\
&-\frac{1}{\hbar} (A^T\xi)\cdot D_t\xi \chi(t,\xi)-\frac{1}{2\hbar} \xi\cdot M\xi \chi(t,\xi)\,\, .
\end{split}
\end{equation}
With $(A^T\xi)\cdot D_t\xi=\xi \cdot AD_t\xi=\frac{1}{2}\xi\cdot(AD_t+D_tA^T)\xi$ we obtain the two relations
\begin{align}
\pa_t R_t&=AR_t\\
\pa_t D_t&=AD_t+D_tA^T+M\,\, , \label{eq:Dt}
\end{align}
the first one is solved by $R_t=\ue^{t A}$ and the second by $D_t=\int_0^t R_s MR_s^T\, \ud s$. To see this we observe that $D_t=\int_0^t R_s MR_s^T\, \ud s$ satisfies 
\begin{equation}
\begin{split}
AD_t+D_tA^T+M&=\int_0^t \frac{\ud }{\ud s} \big(R(_s MR_s^T\big)\, \ud s+M\\
&=R_tM_tR^T=\pa_t D_t\,\, .
\end{split}
\end{equation}
This finishes the proof of \eqref{eq:time-evo-F} since the solution to the initial value problem is unique. 
\end{proof}

In the following we will abuse notation and denote the matrices $D$ and $M$ and the corresponding quadratic functions $D(\xi)=\xi^TD\xi$ and $M(\xi)=\xi^T M\xi$ be the same letter. Using this notation we can write 
\begin{equation}
  D_t(\xi)=\int_0^t M(R_s^T\xi)\, ds=\int_0^t\sum_{k=1}^K \abs{R_sl_k\cdot \xi}^2\, ds  
\end{equation}
where we have used as well that $M(\xi)=\sum_k\abs{l_k\cdot \xi}^2$.

If we insert \eqref{eq:time-evo-F} into \eqref{eq:HS-chi}  we get a Gaussian integral which we can compute and obtain the following result. 

\begin{lem}\label{lem:coh_state-decoherence}
Suppose  $\hat\rho_0=| z_1 \rangle\langle z_2|$, then the time evolved  $\rho_t$ satisfies 
\begin{equation}\label{eq:HS-est}
\norm{\hat \rho_t}_{HS}^2 =\frac{\abs{\det R_{t}}}{\sqrt{\det (I+2C_t)}} \, \ue^{-\frac{1}{\hbar} (\Omega \delta z)\tilde C_t \Omega \delta z} \,\, , 
\end{equation}
where $\delta z=z_1-z_2$ and 
\begin{equation}
\tilde C_t=C_t(I+2C_t)^{-1}\,\, , 
\end{equation}
with $C_t=R_{-t}D_tR_{-t}^T$. 
\end{lem}

\section{H{\"o}rmander Condition}\label{sec:Horm}

By  Lemma \ref{lem:coh_state-decoherence}  we know that a system displays decoherence in phase space if, and only if, the matrix $D_t$ is strictly positive for $t>0$. In this section we will show how this property is related to 
the H{\"o}rmander condition.

Let us start by evaluating the H{\"o}rmander condition for the set of vectorfields 
\begin{equation}
X_0,X_1, \cdots X_{2K}
\end{equation}
defined in \eqref{eq:def_X0} and \eqref{eq:def_Xk} in the form  \eqref{eq:X0_lk}.

\begin{lem}\label{lem:Hoerm_quadr}
Let $V_k$ be the subspaces defined by \eqref{eq:def-Vk}, then we have
\begin{equation}
V_k=V_0+FV_0+F^2V_0+\cdots +F^k V_0\subset \R^{2n}\,\, ,
\end{equation}
for $k=1,2, \cdots$, and  the vector fields $X_0,X_1, \cdots, X_{2K}$ satisfy  H{\"o}rmander's condition if and only if there is a $r\leq 2n-1$ such that $V_r=\R^{2n}$. 
\end{lem}

\begin{proof}
Since constant vector fields commute, we only have to consider $[X_0,X_j]$, and a direct calculation gives for $k=1,2, \cdots , K$
that $[X_0,X_k]\rho=(A\Re l_k)\cdot \nabla \rho$ and $[X_0,X_{K+k}]\rho=(A\Im l_k)\cdot \nabla \rho$, 
which are again constant vector fields, and hence commute among themselves and with the $X_j$, $j\geq 1$. So the only nontrivial commutators to consider are $j$-fold commutators with $X_0$, for $k=1, \cdots , K$ we find 
\begin{align}
[X_0,[X_0,[\cdots , X_k]]]&=A^j \Re l_k\cdot\nabla \,\, ,\\
[X_0,[X_0,[\cdots , Y_{K+k}]]]&=A^k \Im l_k\cdot \nabla\,\, .
\end{align}
Hence we  find $V_1=V_0+AV_0$, $V_2=V_0+AV_0+A^2V_0$, and so on. But $A=F+N\Omega$ and the image of $N$ is contained in $V_0$, 
hence $AV_0\subset V_0+FV_0$ which gives $V_1=V_0+FV_0$, repeating this argument gives $V_k=V_0+FV_0+\cdots F^kV_0$. H{\"o}rmander's condition is now equivalent to $V_k=\R^{2n}$ for 
some $k$. But by the Cayley-Hamilton Theorem $F^{2n+r}$ for $r\geq 0$ can be expressed as a polynomial in $F$ of order $2n-1$, hence $V_k=V_{2n-1}$ for all $k\geq 2n$. 
\end{proof}

Notice that $D_t(\xi)$ is monotonically increasing, i.e., if $t\geq s\geq 0$ then $D_t(\xi)\geq D_s(\xi)\geq 0$, i.e., if it is non-degenerate for $s$, then it is non-degenerate for all $t\geq s$. Therefore we will analyse the Taylor expansion of $D_t(\xi)$ around $t=0$ and check if the first non-zero term is positive. This will then imply positivity of $D_t(\xi)$ for all $t>0$. 

The following result is implicitly contained in \cite{Hor67}, but the first explicit proof seems  to be contained in \cite{Kup72}, see as well \cite{LancPol94}. 

\begin{thm} \label{thm:decoherence} 
The quadratic form $D_t(\xi)$ is nondgenerate for all $t>0$ if, and only if,  the 
H{\"o}rmander condition $V_{2n-1}=\R^{2n}$ holds.

Furthermore, we have for $\xi\notin V_0^{\perp}$ that $M(\xi)\neq 0$ and 
\begin{equation}\label{eq:decoher-0}
D_t(\xi)=M(\xi) t+O(t^2)\,\, , 
\end{equation}
and for any $k\geq 1$ with $V_{k-1}^{\perp}\neq \{0\}$ that for $\xi\in V_{k-1}^{\perp}$ 
\begin{equation}\label{eq:decoher-k}
D_t(\xi)=\frac{M((F^T)^k\xi)}{(2k+1)(k!)^2}t^{2k+1} +O(t^{2k+2})\,\, .
\end{equation}
where $M(\xi)=\sum_{m}\abs{L_m(\Omega^T\xi)}^2$. 
\end{thm}

The strategy of the proof is to study the Taylor expansion of $D_t(\xi)$, as a function of $t$ around $t=0$. 
We will show that for any $\xi\neq 0$ the first non-zero term in the Taylor expansion is positive, and this implies by the monotonicity 
of $D_t(\xi)$ that the quadratic form is non-degenerate and positive for all $t>0$. 

The first step is provided by the following Lemma. 

\begin{lem}\label{lem:derivative-exp} We have 
\begin{equation}
D_0=0\,\, ,
\end{equation}
and for any $j\in \N_0$ 
\begin{equation}\label{eq:D-der}
\pa_t^{j+1}D_t|_{t=0}=\sum_{l=0}^j \binom{j}{l} A^{j-l} M(A^T)^l \,\, .
\end{equation}
\end{lem}

\begin{proof}
We have  derived in Theorem \ref{thm:time-evo} a formula for $D_t$ as an integral, this immediately implies that $D_0=0$. For the study of the derivatives of $D_t$ it is easier to use \eqref{eq:Dt} and its derivatives
\begin{align}
    \pa_t D&=AD+DA^T+M\,\, ,\\
    \pa_t^{j+1}D&=A\pa_t^jD+\pa_t^j DA^T \quad j=0,1,2, \cdots .
\end{align}
With $D_0=0$ these relations immediately give $\pa_tD_0=M$ and \eqref{eq:D-der} is then easily proved by induction. 
\end{proof}

The next Lemma provides a technical step to evaluate $\pa_t^jD_t(\xi)$ for $\xi\in V_{k}^{\perp}$.

\begin{lem}\label{lem:aux-zero} Let $k\geq 1$ and  $\xi\in V_{k-1}^{\perp}$, then for $j=0,1, \cdots , k-1$ 
\begin{equation}
\xi^T A^jM=0\,\, 
\end{equation}
and 
\begin{equation}
\xi^T A^kM=\xi^T F^k M\,\, 
\end{equation}
\end{lem}

\begin{proof}
Notice that 
\begin{equation} 
V_{k-1}^{\perp}=V_0^{\perp}\cap (FV_0)^{\perp}\cap \cdots \cap (F^{k-1}V_0)^{\perp}\,\, , 
\end{equation}
and hence $V_j^{\perp} \subset V_{k-1}^{\perp}$ for $j\leq k-1$. As $\range M=V_0$, we find, as in the proof of Lemma \ref{lem:Hoerm_quadr}, that 
$\range A^j M\subset V_j$ and therefore $\xi\in V_{k-1}^{\perp}$ implies $\xi^T A^j M=0$.

To prove the second relation we notice that  by \eqref{eq:def-N} we have $\ker N=V_0^{\perp}$ and therefore $N\xi=0$. That implies $A^T\xi=F^T\xi+\Omega N\xi=F\xi$ 

Using that $A^k=(F+N\Omega)^k$ and that  $\range (F^j N)\subset V_{k-1}$ for $j\leq k-1$, which implies $\xi^TF^j N=0$, we find 
$\xi^TA^k=\xi^T F^k$, as all other terms obtained by expanding $(F+N\Omega)^k$ are of the form $F^j N B$ for some $j\leq k-1$, where $B$ is product of terms involving $F$ and $\Omega N$. 
\end{proof}

Now we combine the previous two Lemmas.

\begin{lem}\label{lem:Taylor}
\begin{itemize}
\item[(i)] Suppose $\xi\in V_{k-1}^{\perp}$, then for all $j\leq 2k$ 
\begin{equation}
\pa_t^j D_t(\xi)|_{t=0}=0
\end{equation}
and 
\begin{equation}\label{eq:odd-der}
\pa_t^{2k+1}D_t(\xi)|_{t=0}=\binom{2k}{k} \xi^T F^kM(F^T)^k\xi\,\, .
\end{equation}
\item[(ii)] Suppose for all $j\leq 2k+1$ we have 
\begin{equation}
\pa_t^j D_t(\xi)|_{t=0}=0\,\, ,
\end{equation}
then $\xi\in V_k^{\perp}$. 
\end{itemize}
\end{lem}

\begin{proof} 
Part $(i)$ follows directly by combining Lemma \ref{lem:derivative-exp} and Lemma \ref{lem:aux-zero}.

Part $(ii)$ we prove by induction over $k$. Suppose $k=0$, then we have to consider 
\begin{equation}
D_0(\xi)=0 \quad\text{and}\quad \pa_t D_0(\xi)=\xi^TM\xi\,\, , 
\end{equation}
which follow from Lemma \ref{lem:derivative-exp}. And as $\xi^TM\xi=\sum_j (\xi^T \Re l_j)^2+(\xi^T\Im l_j)^2$ we see that $D_0(\xi)=0$ implies $\xi\in V_0^{\perp}$.

Now suppose the assertion holds for $k-1$, then we have to show that if $\xi\in V_{k-1}^{\perp}$, then the further  conditions $\pa_t^{2k}D_0(\xi)=\pa_t^{2k+1}D_0(\xi)=0$ 
imply that $\xi\in (F^k V_0)^{\perp}$, and hence $\xi\in V_k^{\perp}=V_{k-1}^{\perp}\cap (F^k V_0)^{\perp}$. But from Lemma \ref{lem:derivative-exp} and Lemma \ref{lem:aux-zero} we get 
for $\xi\in V_{k-1}^{\perp}$ that $\pa_t^{2k}D_0(\xi)=0$ and from part $(i)$ we have for $\xi\in V_{k-1}^{\perp}$ that 
\begin{equation}
\begin{split}
\pa_t^{2k+1}D_0(\xi)&=\binom{2k}{k} \xi^T F^kM(F^T)^k\xi\\
&=\binom{2k}{k}\sum_j (\xi^T F^k\Re l_j)^2+(\xi^TF^k\Im l_j)^2 \,\, .
\end{split}
\end{equation}
Therefore $\pa_t^{2k+1}D_0(\xi)=0$ implies that $\xi\in (F^k V_0)^{\perp}$, and hence $\xi\in  V_k^{\perp}$. 
\end{proof}

Lemma \ref{lem:Taylor} is the main ingredient in the proof of Theorem \ref{thm:decoherence} which we can now provide. 

\begin{proof}[Proof of Theorem \ref{thm:decoherence}]
Suppose H{\"o}rmander's condition holds, and assume there is a $\xi\in \R^{2n}$ such that $D_{t_0}(\xi)=0$ for some $t_0>0$, and hence by monotonicity 
$D_t(\xi)=0$ for all $t\in [0,t_0]$. This implies $\pa_t^j D_0(\xi)=0$ for all $j\in \N$, and hence $\xi\in V_k^{\perp}$ for all $k\in \N$. But since H{\"o}rmander's condition holds 
$V_{2n-1}^{\perp}=\{0\}$, and so $\xi=0$. Therefore $D_t(\xi)$ is non-degenerate, and as $D_t(\xi)\geq 0$, it follows that $D_t(\xi)> 0$ for $\xi\neq 0$ and $t>0$. 

To prove  that the  non-degeneracy of $D_t(\xi)$ implies that H{\"o}rmander's condition holds, we show that for all $t>0$ 
\begin{equation}
V_{2n-1}^{\perp}\subset \ker D_t(\xi)\,\, .
\end{equation}
Since by Lemma \ref{lem:Hoerm_quadr} the H{\"o}rmander condition is equivalent to $V_{2n-1}^{\perp}=\{0\}$ this implies that if the  the H{\"o}rmander condition does not hold, then 
$D_t(\xi)$ cannot be non-degenerate for $t\geq 0$.  We recall, as we pointed out in the proof of Lemma \ref{lem:Hoerm_quadr}, that by the Cayley Hamilton Theorem we have 
$V_k=V_{2n-1}$ for all $k\geq 2n$, end hence if $\xi\in V_{2n-1}^{\perp}$, then $\xi\in V_k^{\perp}$ for all $k\in \N_0$. Part $(i)$ of Lemma \ref{lem:Taylor} then gives 
\begin{equation}
\pa_t^k D_0(\xi)=0\,\, ,\quad\text{for all}\quad k\in N_0\,\, ,
\end{equation}
and as $D_t(\xi)$ is by construction an analytic function of $t$ this implies that $D_t(\xi)=0$ for all $t\geq 0$ and $\xi\in V_{2n-1}^{\perp}$. 

The formula for the leading order term in the Taylor expansion of $D_t(\xi)$ follows then directly from part $(i)$ of Lemma \ref{lem:Taylor}, if $\xi\in V_{k-1}^{\perp}$ then
\begin{equation}
    D_t(\xi)=\frac{t^{2k+1}}{(2k+1)!}\pa_t^{2k+1}D_t(\xi)|_{t=0}+O(t^{2k+2})
\end{equation}
and now we can insert \eqref{eq:odd-der} to obtain \eqref{eq:decoher-0} and \eqref{eq:decoher-k}. 

\end{proof}

Let us recall the expression for $D_t(\xi)=\xi\cdot D_t\xi$ 
\begin{equation}
D_t(\xi)=\sum_{k=1}^K \int_0^t \abs{\xi\cdot R_sl_k}^2\,  \, \ud s\,\, ,
\end{equation}
and similarly we have for $C_t(\xi)=D_t(R_-t^T\xi)$ 
\begin{equation}\label{eq:C-t-int}
C_t(\xi)=\sum_{k=1}^K\int_0^t  \abs{\xi\cdot R_{-s}l_k}^2\,  \, \ud s\,\, , 
\end{equation}
which follows by substituting $s\to t-s$. But if we substitute $s\to -s$ in \eqref{eq:C-t-int}  we find that  
\begin{equation}\label{eq:C-D}
    C_t(\xi)=-D_{-t}(\xi)\,\, .
\end{equation}
As the leading order terms of $D_t(\xi)$ for small $t$ contain only odd powers of $t$
this implies that they agree with the leading order terms of $C_t(\xi)$, i.e., we have 
\begin{equation}
C_t(\xi)=D_t(\xi)(1+O(t)) \,\, .   
\end{equation}

In the estimates of the rate of decoherence, Lemma \ref{lem:coh_state-decoherence}, contains the quadratic form 
\begin{equation}
 \tilde{C}_t(\xi)=\xi\cdot C_t(I+2C_t)^{-1}\xi\,\, , 
\end{equation}
and the methods from the proof of Theorem \ref{thm:decoherence} can be used as well to find the leading order behaviour for small $t$ of 
$\tilde C_t(\xi)$. 

\begin{prop} \label{prop} We have 
\begin{equation}
\tilde C_t(\xi)=C_t(\xi)(1+O(t))
\end{equation}
and in particular if $\xi\in V_{j-1}^{\perp}$ then
\begin{equation}\label{eq:Dtj}
\tilde C_t(\xi)=\frac{1}{(2j+1) (j!)^2} \sum_{k=1}^K \abs{\xi \cdot F^jl_k}^2\,\, t^{2j+1}+O(t^{2j+2}) \,\, . 
\end{equation}
\end{prop}

\begin{proof}
As $C_t=O(t)$, we can expand $\tilde C_t$ for small $t$ as
\begin{equation}
\tilde C_t=\sum_{n=0}^{\infty}  (-1)^{n}C_t^{n+1}\,\, .
\end{equation}
Now \eqref{eq:decoher-k} applied to the $n=0$ term in the sum  gives the first term in \eqref{eq:Dtj} and it remains to show that the terms 
$C_t^{n+1}$, are of higher order in $t$ for $n=1,2, \cdots $. So let us assume $\xi\in V_{j-1}^{\perp}$, then we have by \eqref{eq:C-t-int} that for any $\eta\in \R^{2n}$  
\begin{equation}
\eta \cdot C_t\xi=\sum_{k}\int_0^t\la \eta , R_{-s}l_k\ra\la \xi , R_{-s}l_k\ra\, \ud s=O(t
\end{equation}
Now we use that by Lemma \ref{lem:aux-zero} 
we have 
\begin{equation}
\la \xi , R_{-s}l_k\ra=\sum_{k=0}^{\infty} \frac{(-s)^k}{k!}\la \xi , A^kl_k\ra=O(s^{j})
\end{equation}
and with $\la \eta , R_{-s}l_k\ra=O(1)$ this gives $\eta \cdot C_t \xi=O(t^{j+1})$. 
As $\eta$ was arbitrary this means that $C_t\xi=O(t^{j+1})$ and together with $C_t=O(t)$ we obtain   
\begin{equation}
\xi  \cdot C_t^{n+1}\xi=(C_t\xi)\cdot C_t^{n-1}C_t\xi=O(t^{2j+1+n})\,\, .
\end{equation}
\end{proof}

If we combine Lemma \ref{lem:coh_state-decoherence} and Proposition \ref{prop} we immediately get Theorem \ref{thm:1} since 
\begin{equation}
   \frac{\abs{\det R_{t}}}{\sqrt{\det (I+2C_t)}}=1+O(t)\,\, , 
\end{equation}
and 
\begin{equation}
    |\Omega \delta z\cdot F^jl_k|^2=\abs{L_k(F^j\delta z}^2\,\, ,
\end{equation}
where we used in addition that $\Omega^T(F^T)^j\Omega =(-1)^jF^j$ which follows from $F=\Omega Q$ and $\Omega^T=\Omega^{-1}=-\Omega$. 

\section{Examples}\label{sec:examples}

In this section we will look at some examples to illustrate our results. 
It will be useful to rewrite the nested sequence of subspaces $V_0\subset V_1\subset \cdots \subset V_k\subset \cdots$ introduced in Definition \ref{def:Hormander}. Let us set $W_0:=V_0$ and let $W_k$ be the orthogonal complement of $V_{k-1}$ in $V_k$, so that $V_{k}=V_{k-1}\oplus W_k$,
 then we have 
\begin{equation}
V_k=W_0\oplus W_1\oplus \cdots \oplus W_k\,\, \quad \text{and for $k\neq j$}\quad W_{k}\perp W_{j}  \,\, .
\end{equation}
If the H{\"o}rmander condition does not hold, then there is a smallest $r$ such that $V_{r+1}=V_r$ and $V_r\neq \R^{2n}$, and we define $W_{DF}$ to be the orthogonal complement of $V_r$ in $\R^{2n}$, so that  
\begin{equation}\label{eq:deco-decomp}
\R^{2n}=W_0\oplus W_1\oplus\cdots \oplus W_r\oplus W_{DF}\,\, .
\end{equation}
Here $W_{DF}$ is the decoherence free subspace, and we can formally include the case that the H{\"o}rmander condition holds as we have $W_{DF}=\{0\}$ then.

Let us first look a free particle with collisional decoherence, i.e., $H=\frac{1}{2m}\hat p^2$ and $L=\sqrt{\Lambda}\, \hat q$. Then we find $N=0$ and $M=\Lambda\, e_pe_p^T$, where $e_p=(1,0)$ is the unit vector in $p$ direction, and $R_t=\begin{pmatrix} 1 & 0\\ t/m & 1\end{pmatrix}$ and so using \eqref{eq:C-t-int} we find
\begin{equation}
    D_t(\xi)=\Lambda t\xi_p^2+\frac{\Lambda}{m} t^2\xi_p\xi_q+\frac{\Lambda}{3m^2} t^3\, \xi_q^2\,\, .
\end{equation}
We see that for small $t$ $D_t$ grows linear in $t$ if $\xi_p\neq 0$, but if $\xi_p=0$ and $\xi_q\neq 0$ then it grows with $t^3$, which means it is smaller for small $t$ and hence the onset of decoherence is slower. This behavior has been observed in \cite{PabHabZur93,Zur03} when comparing the suppression of oscillatory terms in the Wignerfunction of a cat-state, as we discussed already at the end of Section \ref{sec:Introduction}.

Now consider a Hamiltonian with potential, $H=\frac{1}{2m}\hat p^2+V(q)$, and $L=\sqrt{\Lambda}\, \hat q$, then we find $[X_0,X_1]=\sqrt{\Lambda}/m\, \pa_q$, and so in this case the H{\"o}rmander condition still holds, independently of the potential. The matrix $R_t$ will of course depend on the potential, but with $Fe_p=(1/m)e_q$  by Theorem \ref{thm:decoherence} the leading order behavior of $D_t(\xi)$, and therefore the onset of decoherence,  does not depend on the potential $V$.

Let us now look at how the damped harmonic oscillator in one-degree of freedom fits into our scheme, see \cite{BrePet02}. In that situation we have $L_1=\sqrt{\gamma (\bar n+1)}\, a$ and $L_2=\sqrt{\gamma \bar n}\, a^{\dagger}$ where $a=\frac{1}{\sqrt{2}}(\hat q+\ui \hat p)$ is the annihilation operator and $a^{\dagger}$ the corresponding creation operator. Here $\gamma>0$ is a  dissipation constant and $\bar n=(\ue^{\hbar\omega \beta}-1)^{-1}$ is related to the temperature of the bath and $\omega$ is the frequency of the oscillator $H=\hbar \omega (a^{\dagger}a+1/2)$. Then we find $V_0=\R^2$, so  the H{\"o}rmander condition is fulfilled without the need for any commutators, as expected, and we find 
\begin{equation}
M=\frac{\gamma}{2}\coth(\hbar\omega\beta/2)I_2\,\, , \quad 
A=\omega \Omega_2-\frac{\gamma}{2} I_2
\end{equation}
and this leads to 
\begin{equation}
    R_t=\ue^{-\frac{\gamma}{2}\, t}\begin{pmatrix} \cos \omega t) & -\sin(\omega t)\\ \sin(\omega t) & \cos(\omega t)\end{pmatrix} 
\end{equation}
and 
\begin{equation}
    D_t(\xi)=\frac{1}{2}(1-\ue^{-\gamma t})\coth(\hbar\omega\beta/2)\abs{\xi}^2\,\, .
\end{equation}
With $(1-\ue^{-\gamma t})=\gamma t+O(t^2)$ this matches the prediction from Theorem \ref{thm:decoherence}. 

We already mentioned the case $H=\lambda \hat p\hat q$ and $L=\sqrt{\Lambda}\, \hat q$ where the H{\"o}rmander condition is not fulfilled. For this case we find 
\begin{equation}
    D_t(\xi)=\frac{\Lambda}{2\lambda} (1-\ue^{-2\lambda t})\,  \xi_p^2
\end{equation}
and $R_t=\begin{pmatrix} \ue^{-\lambda t} & 0 \\ 0 & \ue^{\lambda t}\end{pmatrix}$, so $D_t(\xi)$ is degenerate for all $t\geq 0$ and hence decoherence does not hold if $\delta q=0$.

We now look at a class of coupled harmonic oscillators with Hamiltonian 
\begin{equation}
\hat H= \sum_i^n\omega _i (a_i^{\dagger}a_i+1/2)+ \sum_{ i\neq j}\delta_{i,j} (a_i^{\dagger}a_j+a_{j}^{\dagger}a_i)
\end{equation}
where $a_i=\frac{1}{\sqrt 2} (\hat q_i+\ui \hat p_i)$, $a_i^*=\frac{1}{\sqrt 2} (\hat q_i-\ui \hat p_i)$ denote the creation and annihilation operates of the $i$'th oscillator, respectively, and the symmetric matrix $\Delta=(\delta_{ij})$, with $\delta_{ii}=0$ and $\delta_{ij}\geq 0$,  defines the coupling of the oscillators. We will write the phase space variables as 
\begin{equation}
   (p_1,q_1,p_2,q_2, \cdots , p_n,q_n)\subset \R^2_1\oplus\R_2^2\oplus \cdots \oplus \R_n^2
\end{equation}
and with this splitting we can write the symplectic form as $\Omega=I_n\otimes \Omega_2$, where $I_n$ is the $n\times n$ identity matrix and $\Omega_2$ is the symplectic form on $\R^2$. Similarly we have $Q=Q_n\otimes I_2$ where 
\begin{equation}
    Q_n=\begin{pmatrix} \omega_1  &  0 & \cdots & 0\\ 0  & \omega_2  & \cdots & 0\\ \vdots & \vdots & \ddots  & \vdots \\ 0 & 0 & \cdots & \omega_n \end{pmatrix}+\Delta
\end{equation}
is an $n\times n$ matrix. For these examples we choose to couple one oscillator of the system, the $k$'th one, to a thermal bath by choosing 
$L_1^{(k)}=\sqrt{\gamma (\bar n_k+1)}\, a_k $ and $L_2^{(k)}=\sqrt{\gamma \bar n_k}\, a_k^{\dagger}$, where $\bar n_k=(\ue^{\hbar\omega_k \beta}-1)^{-1}$ as in the previous example. We are interested how the noise will spread through the system, and if it reaches all parts of it. 

The first case we consider is a chain of $n$ coupled oscillators with thermal noise coupled to the first oscillator, $k=1$. In this case $\Delta$ is tri-diagonal with 
\begin{equation}
    \delta_{ij}=\begin{cases} \delta & \abs{i-j}=1\\ 0 & \text{otherwise}\end{cases}
\end{equation}
and $V_0=\R^2_1$, and we find $V_k=\R^2_1\oplus \R^2_2\oplus\cdots\oplus \R_{k+1}^2$, and so $V_{n-1}$ is the whole phase space and H{\"o}rmander's condition holds. If $\xi\in \R_k^2$ then we have  
\begin{equation}
D_t(\xi)=\frac{\gamma\coth(\hbar\omega_1\beta)}{2(2k+1)(k!)^2} \,\delta^k \abs{\xi}^2\, t^{2k+1}(1+O(t))\,\, ,
\end{equation}
so the onset of decoherence is delayed.

In our next example we consider a chain of three harmonic oscillators and couple the thermal noise to the middle one. Then $V_0=\R_2^2$ and 
\begin{equation}
    V_1=V_0+FV_0=\{(x,y,x)\, ;\, x,y\in\R^2\}
\end{equation}
and $V_2=V_1$, so the H{\"o}rmander condition does not hold. In the notation from \eqref{eq:deco-decomp} we set $W_0=V_0$ and $W_1=\{(x,0,x)\,;\, x\in \R^2\}$ and 
$W_{DF}:=\{(x,0,-x)\, ;\, x\in\R^2\}$, the $V_1=W_0\oplus W_1$ and the phase space is decomposed into
\begin{equation}
    W_1\oplus W_2\oplus W_{DF}
\end{equation}
where $W_{DF}$ stands for the Decoherence free susbspace. Notice that 
$W_1\oplus W_2$ and $W_{DF}$ are symplectic subspaces, and that means there exists a proper subsystem which is decoherence free. This subsystem consists of the oscillator one and three states which are anti-symmetric under permutation of oscillator one and three. 

Finally we look at the example of a chain of four oscillators and we couple the noise to oscillator two. In this situation we find $W_0=V_0=\R^2_2$, $W_2=\{(x,0,x,0)\, ;\, x\in\R^2\}$, $W_3=\R_4^2$ and $W_4=\{(x,0,-x,0)\, ;\, x\in\R^2\}$ so that $V_1=W_0\oplus W_1$, $V_2=W_0\oplus W_1\oplus W_2$ and $V_3=W_0\oplus W_1\oplus W_2\oplus W_3$ is the whole phase space, so the H{\"o}rmander condition holds. Notice how $W_2$ is the same space which appeared in the previous example as the decoherence free subspace $W_{DF}$. 

The difference between the two previous examples is the presence of a symmetry which meant that certain states are affected identically by the noise and therefore superposition of such states don't show decoherence relative to each other. This type of mechanism is a well known as a tool to create decoherence free subspaces and subsystems, see \cite{LidChuWha98,BraBraHaa00,LidWha03,lid14}.

\section{Summary and Outlook}\label{sec:summary}

We considered quantum systems coupled to an environment, and in particular situations where only some degrees of freedom of the system are coupled to the environment. Then the question arises if the influence of the environment will affect eventually the whole system, and how long it takes for noise to start affecting different parts of the system. We showed that for systems described by Gaussian Channels $\mathcal{V}_t$ generated by the GKLS equation, this questions can be addressed using the so called H{\"o}rmander condition, a condition on the commutators of the classical Hamiltonian vectorfields of the internal Hamiltonian and the Lindblad terms. 
Our main result is a condition on when the noise will affect eventually all of the system, and the identification of a set of time-scales for the onset of decoherence in different parts of the system
\begin{equation}
\norm{\mathcal{V}_t(|z\ra\la z'\ra)}\sim \ue^{-\frac{1}{\hbar}d_{j}(z-z')\, t^{2j+1}}\,\, ,     
\end{equation}
where $j=0,1,2,3, \ldots$ .

We then considered a couple of examples which suggested as well that the H{\"o}rmander condition can be used to identify decoherence free subsystems and subspaces as described in \cite{LidChuWha98,lid14,Yam14}. It would be worthwhile to follow this up and explore the decomposition of phase-space into subspaces and how this might translate into a   corresponding decomposition of the Quantum system into subsystems and  subspaces. 

We expect that the results can be extended to more general systems, where the Hamiltonian and the Lindblad operators don't have to be quadratic or linear functions of $\hat x$, respectively, by using the methods developed in \cite{Graefe_2018}. Equations of the form \eqref{eq:Lindblad-ps}, the phase space representation of the Lindblad equation,  with the assumption that the H{\"o}rmander condition holds have been studies extensively in the mathematics literature. One of the main techniques developed is to model these equations on nilpotent Lie groups whose Lie algebra is related to the commutators in the H{\"o}rmander condition, and this in turn is related to sub-Riemannian geometry, see \cite{Bram14,AgrBarBos20}. We expect that these tools will prove to be fruitful in the study of decoherence, too.

\begin{acknowledgments}
R.S. acknowledges the financial support provided by the
EPSRC Grant No. EP/P021123/1.
\end{acknowledgments}


\bibliography{open}

\end{document}